\documentclass[12pt]{article}
\usepackage[paper=a4paper,margin=1in]{geometry}
\usepackage{t1enc}      % needed to have the edth and thorn operators

\usepackage{amsmath}
\usepackage{amsfonts}
\usepackage{amssymb}
\usepackage{amsthm}
\usepackage{graphicx}
\usepackage{mathrsfs}  % needed to have the correct character for scri

\newcommand{\uA}{\underline A \,}
\newcommand{\uB}{\underline B \,}
\newcommand{\uC}{\underline C \,}
\newcommand{\uD}{\underline D \,}

\newcommand{\bi}{\bf i}
\newcommand{\bj}{\bf j}
\newcommand{\bk}{\bf k}
\newcommand{\bl}{\bf l}
\newcommand{\bm}{\bf m}
\newcommand{\bn}{\bf n}

\newtheorem{theorem}{Theorem}[section]

\numberwithin{equation}{section}

\begin{document}
\bibliographystyle{unsrt}

\title{On the total mass of closed universes with a positive cosmological 
constant}
\author{L\'aszl\'o B. Szabados \\
Wigner Research Centre for Physics, \\
H-1525 Budapest 114, P. O. Box 49, Hungary, \\
E-mail: lbszab@rmki.kfki.hu}

\maketitle

\begin{abstract}
The recently suggested notion of total mass density for closed universes is 
extended to closed universes with a positive cosmological constant. Assuming 
that the matter fields satisfy the dominant energy condition, it is shown 
that the cosmological constant provides a sharp lower bound for the total 
mass density, and that the total mass density takes this as its minimum value 
if and only if the spacetime is \emph{locally} isometric with the de Sitter 
spacetime. This notion of total mass density is extensible to 
\emph{non-compact} three-spaces of homogeneity of Bianchi class A cosmological 
spacetimes. 
\end{abstract}

%%%%%%%%%%%%%%%%%%%%%%%%%%%%%%%%%%%%%%%%%%%%%%%%%%%%%%%%%%%%%%%

\section{Introduction}
\label{sec-1}

In our previous paper \cite{Sz12} we suggested a notion of total mass ${\tt 
M}$ (or rather the total mass density, depending on the normalization) of 
closed universes at any instant represented by a closed spacelike hypersurface 
$\Sigma$. This ${\tt M}$ was defined by a non-negative expression, built from 
the integral of the 3-surface twistor operator and the energy-momentum tensor 
on $\Sigma$, which in the asymptotically flat/hyperboloidal case provided a 
lower bound for the ADM/Bondi--Sachs mass. It was shown in \cite{Sz12} that, 
apart from a numerical coefficient, ${\tt M}$ is just the first eigenvalue of 
the square of the Sen--Witten operator, Witten's gauge condition admits a 
non-trivial solution if and only of ${\tt M}=0$, and that ${\tt M}=0$ holds 
on some (and hence any) Cauchy hypersurface $\Sigma$ if and only if the 
spacetime is flat with toroidal spatial topology: $\Sigma\approx S^1\times S^1
\times S^1$. If we allow that the connection not to be determined completely 
by the metric, namely if locally flat but holonomically non-trivial spacetime 
configurations are allowed, then this theorem should be generalized. This 
generalization was done in \cite{Sz13}: ${\tt M}=0$ holds if and only if the 
spacetime is holonomically trivial. In \cite{Sz13} we discussed the properties 
of ${\tt M}$ further (and changed its normalization such that its scale be 
that of the ADM/Bondi--Sachs mass). We showed that the multiplicity of each 
eigenvalue of the square of the Sen--Witten operator is always even. We also 
illustrated these ideas in the examples of the closed Bianchi I. and 
Friedman--Robertson--Walker (FRW) cosmological models. 

In the present paper we extend the ideas above by allowing the presence of a 
\emph{positive} cosmological constant in Einstein's equations. Its potential 
significance is given by the phenomenological interpretation of the observed 
luminosity-redshift anomaly of distant type Ia supernovae as the indication 
of the \emph{strict positivity} of the cosmological constant $\Lambda$ (see 
e.g. \cite{Pe,GP}). 

While the general idea behind the construction and the key technical (both 
geometric and functional analytic) results remain the same in the presence of 
a positive $\Lambda$, the cosmological constant provides a sharp, strictly 
positive lower bound for ${\tt M}$. Thus the minimal mass configurations will 
be different from that in the zero cosmological constant case above. The aim 
of the present paper is to determine these configurations. We show that ${\tt 
M}$ takes this sharp lower bound as its minimal value precisely when the 
spacetime is \emph{locally} isometric with the de Sitter spacetime. In this 
case ${\tt M}$ turns out to be independent of the hypersurface $\Sigma$. We 
also indicate how this notion of total mass density can be extended to 3-spaces 
of homogeneity of Bianchi class A cosmological spacetimes even if these are 
not closed. 

In the next section we review the key results and formulae that we need in our 
analysis. In Section \ref{sec-3} we determine the minimal mass configurations. 
In section \ref{sec-4} it is indicated how this notion of the total mass 
density can be defined for spatially non-compact Bianchi A cosmological models. 
In the appendices we show, using Aronszajn's theorem, that the eigenspinors of 
the (square of the) Sen--Witten operators cannot vanish on any open subset of 
$\Sigma$; and we discuss the technical details on the geometry of the Bianchi 
cosmological spacetimes that we need in the proof of our main theorem. As a 
by-product, we show that the left-invariant frame fields always satisfy the 
special orthonormal frame gauge condition of Nester \cite{Nester}. 

Here we use the abstract index formalism, and our sign conventions are those 
of \cite{PR}. In particular, the signature of the spacetime metric is 
$(+,-,-,-)$, and the curvature tensor is defined by $-R^a{}_{bcd}X^bV^cW^d
:=V^c\nabla_c(W^d\nabla_dX^a)-W^c\nabla_c(V^d\nabla_dX^a)-[V,W]^c\nabla_cX^a$ 
for any vector fields $X^a$, $V^a$ and $W^a$. Thus, Einstein's equations take 
the form $G_{ab}:=R_{ab}-\frac{1}{2}Rg_{ab}=-\kappa T_{ab}-\Lambda g_{ab}$, where 
$\kappa:=8\pi G$ and $G$ is Newton's gravitational constant, and $\Lambda>0$.

%%%%%%%%%%%%%%%%%%%%%%%%%%%%%%%%%%%%%%%%%%%%%%%%%%%%%%%%%%%%%%%%%%

\section{Preliminaries}
\label{sec-2}

In our investigations the key geometric ingredient is the Reula--Tod (or 
$SL(2,\mathbb{C})$ spinor) form \cite{RT} of the Sen--Witten identity for 
any spinor field $\lambda^A$ on the spacelike hypersurface $\Sigma$: 

\begin{eqnarray}
D_a\!\!\!\!&\Bigl(\!\!\!\!&t^{A'B}\bar\lambda^{B'}{\cal D}_{BB'}\lambda^A-
 \bar\lambda^{A'}t^{AB'}{\cal D}_{B'B}\lambda^B\Bigr)+2t^{AA'}\bigl({\cal D}_{AB'}
 \bar\lambda^{B'}\bigr)\bigl({\cal D}_{A'B}\lambda^B\bigr)\nonumber \\
\!\!\!\!&=\!\!\!\!&-t_{AA'}h^{ef}\bigl({\cal D}_e\lambda^A\bigr)\bigl(
 {\cal D}_f\bar\lambda^{A'}\bigr)-\frac{1}{2}\kappa t^aG_{ab}\lambda^B\bar
 \lambda^{B'}, \label{eq:2.1}
\end{eqnarray}
where $t^a$ is the future pointing unit timelike normal to $\Sigma$, $h_{ab}:=
g_{cd}P^c_aP^d_b$ is the induced (negative definite) metric, $D_e$ is the 
corresponding intrinsic Levi-Civita covariant derivative and ${\cal D}_e:=
P^f_e\nabla_f$ is the derivative operator of the Sen connection. Here $P^a_b:=
\delta^a_b-t^at_b$ is the $g_{ab}$--orthogonal projection to $\Sigma$. The key 
observation is that the algebraically irreducible decomposition of the unitary 
spinor form \cite{Reula,joerg} of the ${\cal D}_e$--derivative of the spinor 
field $\lambda_A$ into its totally symmetric part and the trace(s), 

\begin{equation}
{\cal D}_{EF}\lambda_A={\cal D}_{(EF}\lambda_{A)}+\frac{2\sqrt{2}}{3}
t_F{}^{E'}P^{CC'}_{EE'}\varepsilon_{CA}{\cal D}_{C'D}\lambda^D, 
\label{eq:2.2}
\end{equation}
is a $t_{AA'}$-orthogonal decomposition, and hence it is an $L_2$--orthogonal 
decomposition also. Here the $L_2$-scalar product of two spinor fields, say 
$\lambda^A$ and $\mu^A$, is defined by $\langle\,\lambda^A,\mu^A\,\rangle:=
\int_\Sigma\sqrt{2}t_{AA'}\lambda^A\bar\mu^{A'}$ ${\rm d}\Sigma$, and the 
corresponding norm will be denoted by $\Vert\,.\,\Vert_{L_2}$. In 
(\ref{eq:2.2}) the totally symmetric part of the derivative defines the 
3-surface twistor operator of Tod \cite{Tod84}, while the second term is 
proportional to the action of the Sen--Witten operator (i.e. the Dirac 
operator built from the Sen connection on $\Sigma$) on the spinor field. 

Substituting this decomposition into identity (\ref{eq:2.1}), taking its 
integral on $\Sigma$ and using the fact that $\Sigma$ is closed (i.e. compact 
with no boundary), we obtain 

\begin{equation}
\frac{4\sqrt{2}}{3\kappa}\Vert{\cal D}_{A'A}\lambda^A\Vert^2_{L_2}=
\frac{\sqrt{2}}{\kappa}\Vert{\cal D}_{(AB}\lambda_{C)}\Vert^2_{L_2}+\int_\Sigma 
t^a\bigl(T_{ab}+\frac{\Lambda}{\kappa}g_{ab}\bigr)\lambda^B\bar\lambda^{B'}
{\rm d}\Sigma, \label{eq:2.3}
\end{equation}
where we have used Einstein's equations. This equation will play a key role 
in what follows and we call it the \emph{basic norm identity}. 

Next, let us define 

\begin{equation}
{\tt M}:=\inf\Bigl\{\frac{\sqrt{2}}{\kappa}\Vert{\cal D}_{(AB}\lambda_{C)}
\Vert^2_{L_2}+\int_\Sigma t^a\bigl(T_{ab}+\frac{\Lambda}{\kappa}g_{ab}\bigr)
\lambda^B\bar\lambda^{B'}{\rm d}\Sigma\Bigr\}, \label{eq:2.4}
\end{equation}
where the infimum is taken on the set of the smooth spinor fields for which 
$\Vert\lambda^A\Vert^2_{L_2}=\sqrt{2}$. The physical dimension of ${\tt M}$ is 
mass-density, and note that ${\tt M}$ with this normalization is 
$\sqrt{2}$-times of the ${\tt M}$ introduced in \cite{Sz12}. (Although the 
physical dimension of ${\tt M}$ with the normalization $\Vert\lambda^A\Vert
^2_{L_2}=\sqrt{2}\,{\rm vol}\,(\Sigma)$ would be mass, it is more convenient to 
use the above normalization \cite{Sz13}.) 

By the basic norm identity and the definition of ${\tt M}$ we have, for any 
spinor field $\lambda^A$, that 

\begin{equation}
\frac{4\sqrt{2}}{3\kappa}\Vert{\cal D}_{A'A}\lambda^A\Vert^2_{L_2}=\Bigl\{
\frac{\sqrt{2}}{\kappa}\Vert{\cal D}_{(AB}\lambda_{C)}\Vert^2_{L_2}+\int
_\Sigma t^a\bigl(T_{ab}+\frac{\Lambda}{\kappa}g_{ab}\bigr)\lambda^B\bar\lambda
^{B'}{\rm d}\Sigma\Bigr\}\geq\frac{1}{\sqrt{2}}{\tt M}\Vert\lambda^A\Vert^2
_{L_2}. \label{eq:2.5}
\end{equation}
Since ${\tt M}$ was defined as the infimum of an expression on a set of 
certain smooth spinor fields, it is not \emph{a priori} obvious that there is 
a smooth spinor field which saturates the inequality on the right. 
Nevertheless, one can in fact show that \emph{such a spinor field does 
exist} \cite{Sz12}. We will call such a spinor field a \emph{minimizer} 
spinor field. Thus, if $\lambda^A$ is such a minimizer spinor field, then by 
(\ref{eq:2.5}) we have that 

\begin{equation}
\langle\, 2{\cal D}^{AA'}{\cal D}_{A'B}\lambda^B-\frac{3}{4}\kappa\,{\tt M}
\lambda^A\, ,\,\lambda^A\rangle=2\Vert{\cal D}_{A'A}\lambda^A\Vert^2_{L_2}-
\frac{3}{4}\kappa\,{\tt M}\Vert\lambda^A\Vert^2_{L_2}=0. \label{eq:2.6}
\end{equation}
This implies that the minimizer spinor field is necessarily $L_2$--orthogonal 
to the spinor field $2{\cal D}^{AA'}{\cal D}_{A'B}\lambda^B-\frac{3}{4}\kappa\,
{\tt M}\lambda^A$, or that $\frac{3}{4}\kappa\,{\tt M}$ is an 
\emph{eigenvalue} and the minimizer spinor field is a corresponding 
\emph{eigenspinor} of the operator $2{\cal D}^{AA'}{\cal D}_{A'B}$. 

In fact, it has been proven in \cite{Sz12} that its \emph{smallest} eigenvalue 
is just $\alpha^2_1=\frac{3}{4}\kappa\,{\tt M}$, and the minimizer spinor 
field is a corresponding eigenspinor. Moreover, it was shown in \cite{Sz13} 
that the multiplicity of every eigenvalue is even. In particular, if $\alpha
\not=0$ and $\lambda^A$ is an eigenspinor of $2{\cal D}^{AA'}{\cal D}_{A'B}$ 
with the eigenvalue $\alpha^2$, then $\mu^A:=-{\rm i}\frac{\sqrt{2}}{\alpha}
{\cal D}^A{}_{A'}\bar\lambda^{A'}$ is a \emph{linearly independent} eigenspinor 
with the same eigenvalue. Clearly, all these general results hold true even 
in the presence of a \emph{positive} cosmological constant, since the previous 
analysis can be repeated with the `effective energy-momentum tensor' $\tilde 
T_{ab}:=T_{ab}+\frac{\Lambda}{\kappa}g_{ab}$, which also satisfies the dominant 
energy condition if $T_{ab}$ does.

%%%%%%%%%%%%%%%%%%%%%%%%%%%%%%%%%%%%%%%%%%%%%%%%%%%%%%%%%%%%%%%%%

\section{The minimal mass configurations}
\label{sec-3}

The (geometrical and physical) significance of ${\tt M}$ in the case of the 
vanishing cosmological constant is shown by the result of \cite{Sz12,Sz13} 
that ${\tt M}=0$ on some $\Sigma$ is equivalent to the holonomic triviality 
of the spacetime with spatial topology $\Sigma\approx S^1\times S^1\times S^1$, 
provided the matter fields satisfy the dominant energy condition. Thus, ${\tt 
M}$ satisfies the minimal requirement to be a measure of the `strength of the 
gravitational field'. 

If the cosmological constant is positive, then the minimal mass 
configuration(s) will be different from that in the zero cosmological constant 
case. To see this, let $\lambda^A$ be any spinor field with $\Vert\lambda^A
\Vert^2_{L_2}=\sqrt{2}$, for which we have that 

\begin{equation*}
\frac{\sqrt{2}}{\kappa}\Vert{\cal D}_{(AB}\lambda_{C)}\Vert^2_{L_2}+\int_\Sigma 
t^a\bigl(T_{ab}+\frac{\Lambda}{\kappa}g_{ab}\bigr)\lambda^B\bar\lambda^{B'}
{\rm d}\Sigma\geq\frac{\Lambda}{\kappa}.
\end{equation*}
Then the definition (\ref{eq:2.4}) yields that $\kappa\,{\tt M}\geq\Lambda$, 
i.e. the \emph{total mass density is bounded from below by the cosmological 
constant}. This bound is \emph{sharp}, as the following example shows. 

Let $\tilde\Sigma_t$ be a $t={\rm const}$ spacelike hypersurface of the de 
Sitter spacetime with the line element 

\begin{equation}
ds^2=dt^2-a^2\cosh^2\bigl(\frac{t}{a}\bigr)\Bigl(d\chi^2+\sin^2\chi\,
\bigl(d\theta^2+\sin^2\theta\, d\phi^2\bigr)\Bigr), \label{eq:3.0}
\end{equation}
where $t\in\mathbb{R}$, $(\chi,\theta,\phi)\in S^3$ and $a$ is a positive 
constant. Thus, the induced metric $h_{ab}$ on $\tilde\Sigma_t\approx S^3$ is 
that of the metric sphere with radius $a\cosh(\frac{t}{a})$, its (spatial) 
scalar curvature is ${\cal R}=6/(a^2\cosh^2(\frac{t}{a}))$ and its extrinsic 
curvature is $\chi_{ab}=\frac{1}{a}\tanh(\frac{t}{a})\,h_{ab}$. The spacetime 
is of constant curvature with scalar curvature $R=12/a^2$. Considering this 
spacetime as a solution of the vacuum Einstein equations, for the cosmological 
constant we obtain that $R=4\Lambda$. The first eigenvalue of $2{\cal D}^{AA'}
{\cal D}_{A'B}$ on the $t={\rm const}$ spacelike hypersurfaces of the general 
$k=1$ (closed) FRW spacetimes is $\alpha^2_1=\frac{3}{8}{\cal R}+\frac{1}{4}
\chi^2$, and the two linearly independent eigenspinors are just the ones whose 
components in the globally defined spinor dyad associated with the 
$SU(2)$--left--invariant orthonormal frame are \emph{constant} (see 
\cite{Sz12}, and, for a more general analysis, Appendix \ref{sub-A.2.2}). 
Substituting the specific expression for the spatial curvature scalar and the 
extrinsic curvature here, we obtain that $\kappa\,{\tt M}=\frac{4}{3}\alpha^2_1
=\frac{3}{a^2}=\Lambda$; i.e. the lower bound $\Lambda$ for $\kappa\,{\tt M}$ 
is, in fact, sharp. 

However, the de Sitter spacetime with $S^3$ spatial topology is not the only 
minimal total mass density configuration. In fact, let us replace $\tilde
\Sigma_t$ in (\ref{eq:3.0}) by any homogeneous Riemannian 3-manifold $\Sigma
_t$ which is \emph{locally} isometric with $\tilde\Sigma_t$. It is known (see 
\cite{KN1} Theorem 3 in Note 4, pp. 294-7) that any such $\Sigma_t$ is 
isometric with the quotient $\tilde\Sigma_t/G$, where $G$ is some 
\emph{discrete} subgroup of the group $SU(2)\approx S^3\approx\tilde\Sigma_t$. 
The projection $\pi:\tilde\Sigma_t\rightarrow\Sigma_t$ is clearly a 
\emph{local} isometry, and hence $\Sigma_t$ inherits a globally defined 
left--invariant frame field. In addition, $\pi^*$ maps the extrinsic curvature 
of $\Sigma_t$ into that of $\tilde\Sigma_t$. Then the resulting spacetime is 
\emph{locally} isometric with the de Sitter spacetime, but its global topology 
is $\mathbb{R}\times S^3/G$. For example, $S^3/G$ could be the 3-dimensional 
real projective space $\mathbb{R}P^3:=S^3/\mathbb{Z}_2\approx SO(3)$. 
Therefore, any spinor field on $\Sigma_t$ whose components in the spinor dyad 
associated with the left--invariant frame field are constant is an eigenspinor 
of $2{\cal D}^{AA'}{\cal D}_{A'B}$ with the \emph{same} eigenvalue $\alpha^2_1=
\frac{3}{4}\Lambda$ as in the de Sitter case (see the discussion following 
equations (\ref{eq:A.2.6b}) and (\ref{eq:A.2.10}) of the appendix), yielding 
the same total mass density ${\tt M}=\Lambda/\kappa$, too. 

The next theorem gives the complete characterization of the minimal total 
mass density configurations in the presence of a positive cosmological 
constant.

\begin{theorem}
Let the cosmological constant $\Lambda$ be positive and the matter fields 
satisfy the dominant energy condition. Then $\kappa\,{\tt M}=\Lambda$ for some 
(and hence for any) $\Sigma$ if and only if $T_{ab}=0$ and the spacetime is 
locally isometric with the de Sitter spacetime. \label{th:3.1}
\end{theorem}
\begin{proof}
First we show that $\kappa\,{\tt M}$, calculated on \emph{any} Cauchy surface 
of \emph{any} locally de Sitter spacetime above, is $\Lambda$. We saw that 
$\alpha^2=\frac{3}{4}\Lambda$ is an eigenvalue of $2{\cal D}^{AA'}{\cal D}
_{A'B}$ on any \emph{maximally symmetric} spacelike hypersurface $\Sigma_t$, 
and let $\lambda^A$ be a corresponding eigenspinor. As we noted above, its 
components in the spinor dyad associated with the left--invariant orthonormal 
frame field are constant. Since $\mu_A:=-{\rm i}\frac{\sqrt{2}}{\alpha}{\cal 
D}_{AA'}\bar\lambda^{A'}$ is an independent eigenspinor of $2{\cal D}^{AA'}{\cal 
D}_{A'B}$ with the same eigenvalue (see \cite{Sz13}), its components are also 
constant. Therefore, by equations (\ref{eq:A.2.6b}) and (\ref{eq:A.2.9b}) of 
Appendix \ref{sub-A.2.2}, both $\lambda^A$ and $\mu^A$ satisfy the 3-surface 
twistor equation on $\Sigma_t$, and hence, by (\ref{eq:2.2}) and the 
eigenvalue equation, the equations 

\begin{equation}
{\cal D}_e\lambda^A=\frac{2}{3}P^{AA'}_e{\cal D}_{A'B}\lambda^B=-{\rm i}\alpha
\frac{\sqrt{2}}{3}P^{AA'}_e\bar\mu_{A'}, \hskip 20pt
{\cal D}_e\mu^A=\frac{2}{3}P^{AA'}_e{\cal D}_{A'B}\mu^B={\rm i}\alpha
\frac{\sqrt{2}}{3}P^{AA'}_e\bar\lambda_{A'}, \label{eq:3.1}
\end{equation}
too. In this way, we have a spinor field (or rather a pair of spinor fields) 
on every leaf $\Sigma_t$, and hence on the spacetime manifold itself up to an 
unspecified function of time as a factor of proportionality. To find out how 
to fix this factor, let us suppose that $\omega^A$ is any spinor field on the 
spacetime manifold which satisfies the 3-surface twistor equation on the 
leaves $\Sigma_t$, and decompose its covariant derivative $\nabla_e\omega^A$ 
using ${\cal D}_{(AB}\omega_{C)}=0$ in its equivalent form ${\cal D}_e\omega^A=
\frac{2}{3}P^{AA'}_e{\cal D}_{A'B}\omega^B$. We obtain 

\begin{equation*}
\nabla_e\omega^A-\frac{1}{2}\delta^A_E\nabla_{E'B}\omega^B=\frac{1}{3}\bigl(
2\delta^A_Bt_e-\delta^A_Et_{E'B}\bigr)\Bigl(\frac{3}{2}t^f\nabla_f\omega^B-
t^{BB'}{\cal D}_{B'D}\omega^D\Bigr).
\end{equation*}
Since the vanishing of the left hand side is just the 1--valence twistor 
equation on $M$ \cite{PR}, we obtained that the spinor field $\omega^A$ 
satisfying the 3-surface twistor equation satisfies the 1--valence twistor 
equation if and only if $t^e\nabla_e\omega^A=\frac{2}{3}t^{AA'}{\cal D}_{A'B}
\omega^B$. Since the twistor equation would be the natural extension of the 
3-surface twistor equation form the hypersurfaces to the spacetime, we require 
that the eigenspinors $\lambda^A$ and $\mu^A$ satisfy the `evolution equations' 

\begin{equation}
t^e\bigl(\nabla_e\lambda^A\bigr)={\rm i}\alpha\frac{\sqrt{2}}{3}t^A{}_{B'}
\bar\mu^{B'}, \hskip 20pt
t^e\bigl(\nabla_e\mu^A\bigr)=-{\rm i}\alpha\frac{\sqrt{2}}{3}t^A{}_{B'}\bar
\lambda^{B'}.\label{eq:3.2}
\end{equation}
These evolution equations preserve (\ref{eq:3.1}). In fact, by 
(\ref{eq:3.1}), (\ref{eq:3.2}) and the definition and the specific form of 
the curvature we have that 

\begin{eqnarray*}
t^c\nabla_c\Bigl({\cal D}_e\lambda^A+{\rm i}\alpha\frac{\sqrt{2}}{3}P^{AA'}_e
  \bar\mu_{A'}\Bigr)\!\!\!\!&=\!\!\!\!&-\frac{4}{9}\alpha^2P^{AA'}_et_{A'B}
  \lambda^B-R^A{}_{Bce}t^c\lambda^B=\\
\!\!\!\!&=\!\!\!\!&\frac{1}{6}\bigl(\Lambda-\frac{4}{3}\alpha^2\bigr)\bigl(
  \delta^A_Et_{E'B}+t^A{}_{E'}\varepsilon_{EB}\bigr)\lambda^B=0;
\end{eqnarray*}
and, similarly, the time derivative of ${\cal D}_e\mu^A-{\rm i}\alpha
\frac{\sqrt{2}}{3}P^{AA'}_e\bar\lambda_{A'}$ is also vanishing. Thus the 
integrability conditions of the twistor equation on $M$ are satisfied. 
Therefore, by (\ref{eq:3.2}), there is a unique extension of the eigenspinor 
from e.g. $\Sigma_0$ to the whole locally de Sitter spacetime, which solves 
the 1--valence twistor equation\footnote{In fact, the pair $(\lambda^A,\bar
\mu^{A'})$ of spinor fields solves a pair of 1--valence twistor equations in 
which the secondary part of one solution is just the primary part of the 
other. It might be worth noting that this system of equations can be written 
in terms of Dirac spinors in the remarkably simple form 

\begin{equation*}
\nabla_e\Psi^\alpha=\frac{\rm i}{3}\alpha\,\gamma^\alpha_{e\beta}\Psi^\beta,
\end{equation*}
where $\Psi^\alpha=(\lambda^A,\bar\mu^{A'})$ (as a column vector), $\alpha=
A\oplus A'$, $\beta=B\oplus B'$, and Dirac's `$\gamma$-matrices' are given 
explicitly in the abstract index formalism by 

\begin{equation*}
\gamma^\alpha_{e\beta}=\sqrt{2}\left(\begin{array}{cc}
          0&\varepsilon_{E'B'}\delta^A_E \\
      \varepsilon_{EB}\delta^{A'}_{E'}&0 \end{array}\right) 
\end{equation*}
(see e.g. \cite{PRI}, pp 221). } $\nabla^{A'(A}\lambda^{B)}=0$. 

Let $\Sigma$ be \emph{any} Cauchy hypersurface of the spacetime. Since the 
spinor field $\lambda^A$ solves the 1--valence twistor equation on $M$, its 
restriction to $\Sigma$ solves the 3-surface twistor equation there. Since 
$T_{ab}=0$, for this spinor field the expression between the curly brackets in 
(\ref{eq:2.4}) is just $\Lambda\Vert\lambda^A\Vert^2_{L_2}/(\sqrt{2}\kappa)$. 
Thus, after normalization, it is a minimizer spinor field, yielding that 
$\kappa\,{\tt M}=\Lambda$. 

\medskip
Conversely, let $\Sigma$ be a closed spacelike hypersurface, suppose that 
$\kappa\,{\tt M}=\Lambda$ on $\Sigma$, and let $\lambda^A$ be any eigenspinor 
of $2{\cal D}^{AA'}{\cal D}_{A'B}$ with the eigenvalue $\alpha^2=\frac{3}{4}
\Lambda$. Let $\mu^A:=-{\rm i}\frac{\sqrt{2}}{\alpha}{\cal D}^A{}_{A'}\bar
\lambda^{A'}$, and hence $\mu^A$ is also an eigenspinor of $2{\cal D}^{AA'}
{\cal D}_{A'B}$ with the same eigenvalue. Since $\kappa\,{\tt M}=\Lambda$, by 
the dominant energy condition it follows from (\ref{eq:2.4}) that $t^aT_{aBB'}
\lambda^B\bar\lambda^{B'}=t^aT_{aBB'}\mu^B\bar\mu^{B'}=0$, and hence that $T_{ab}
=f\lambda_A\lambda_B\bar\lambda_{A'}\bar\lambda_{B'}=\tilde f\mu_A\mu_B\bar
\mu_{A'}\bar\mu_{B'}$ for some non-negative functions $f$ and $\tilde f$ on 
$\Sigma$. These imply that if $\lambda^A$ and $\mu^A$ are not proportional 
with each other, then $f=\tilde f=0$. (Note that the \emph{linear 
independence} of $\lambda^A$ and $\mu^A$ means only that $\lambda^A\not=c
\mu^A$ for any non-zero complex \emph{constant} $c$. Thus, in principle, 
$\lambda^A$ and $\mu^A$ could be proportional to each other with some complex 
\emph{function} as a factor of proportionality.) 

Also, the condition $\kappa\,{\tt M}=\Lambda$ yields from (\ref{eq:2.4}) 
that ${\cal D}_{(AB}\lambda_{C)}={\cal D}_{(AB}\mu_{C)}=0$. Substituting these 
into (\ref{eq:2.2}), we obtain (\ref{eq:3.1}). First, we show that $\lambda^A$ 
and $\mu^A$ are not proportional to each other. Thus, suppose, on the contrary, 
that there exists an open set $U\subset\Sigma$ and a smooth complex-valued 
function $F:U\rightarrow\mathbb{C}$ such that ${\cal D}^A{}_{A'}\bar
\lambda^{A'}=F\lambda^A$, i.e. $\mu^A=-{\rm i}\frac{\sqrt{2}}{\alpha}F
\lambda^A$ on $U$. Then, on $U$, $\lambda_A\mu^A=0$. Taking its derivative, 
using equation (\ref{eq:3.1}), our assumption of the proportionality of 
$\mu^A$ and $\lambda^A$ and the eigenvalue equation for $\lambda^A$, we find 

\begin{eqnarray*}
0\!\!\!\!&=\!\!\!\!&D_e(\lambda_A\mu^A)=-\bigl({\cal D}_e\lambda^A\bigr)\mu_A
  +\bigl({\cal D}_e\mu^A\bigr)\lambda_A= \\
\!\!\!\!&=\!\!\!\!&{\rm i}\alpha\frac{\sqrt{2}}{3}P^{AA'}_e\mu_A\bar\mu_{A'}+
  {\rm i}\alpha\frac{\sqrt{2}}{3}P^{AA'}_e\lambda_A\bar\lambda_{A'}={\rm i}
  \frac{\sqrt{2}}{3\alpha}\Bigl(2F\bar F+\alpha^2\Bigr)P_e^{AA'}\lambda_A\bar
  \lambda_{A'};
\end{eqnarray*}
i.e. the projection $Z_e:=P^{AA'}_e\lambda_A\bar\lambda_{A'}$ would have to be 
vanishing on the open set $U$. Since $L^a:=\lambda^A\bar\lambda^{A'}$ is null 
and $\Sigma$ is spacelike, this would imply the vanishing of $\lambda^A$ on 
the open set $U$. However, B\"ar \cite{Bar} showed that the zero-set of any 
eigenspinor of a Riemannian Dirac operator on an $n$-dimensional manifold is 
at most $(n-2)$ dimensional. Also, by Aronszajn's theorem \cite{Aronszajn} 
if a function (or a set of functions) satisfying a second order elliptic 
p.d.e. is (are) vanishing on an open set, then it is (they are) vanishing 
everywhere. These results indicate that $\mu^A$ and $\lambda^A$ cannot be 
proportional. In fact, in Appendix \ref{sub-A.1} we show that the eigenspinors 
of $2{\cal D}^{AA'}{\cal D}_{A'B}$ satisfy the conditions of Aronszajn's theorem, 
and hence they cannot be vanishing on any open subset of $\Sigma$. Therefore, 
$\mu^A$ and $\lambda^A$ are not proportional with each other on any open 
subset of $\Sigma$, and hence, by the argumentation above, \emph{the energy 
momentum tensor is vanishing} on $\Sigma$. 

Next, let us evaluate the integrability conditions of (\ref{eq:3.1}). The 
action of the commutator of two Sen derivative operators on any spinor field 
is 

\begin{equation*}
\bigl({\cal D}_c{\cal D}_d-{\cal D}_d{\cal D}_c\bigr)\lambda^A=-R^A{}_{Bef}
P^e_cP^f_d\lambda^B-2\chi^e{}_{[c}t_{d]}{\cal D}_e\lambda^A,
\end{equation*}
where $\chi_{ab}$ is the extrinsic curvature of $\Sigma$ in the spacetime (see 
e.g. \cite{Sz12}). Thus, it is a straightforward calculation to derive the 
integrability conditions of (\ref{eq:3.1}), for which we obtain 

\begin{equation}
R^A{}_{Bef}P^e_cP^f_d\lambda^B=-\frac{1}{3}\Lambda\delta^A_{(E}\varepsilon_{F)B}
\varepsilon_{E'F'}P^e_cP^f_d\lambda^B, \label{eq:3.3}
\end{equation}
where we used the eigenvalue equation and the specific value $\alpha^2=
\frac{3}{4}\Lambda$ of the eigenvalue. However, the expression on the right 
is just the pull back to $\Sigma$ of the anti-self-dual part of the curvature 
of the constant positive curvature spacetime with scalar curvature $R=4
\Lambda$, i.e. of the de Sitter spacetime. Therefore, the anti-self-dual 
part of the curvature tensor of the spacetime has the structure $-R_{ABCC'DD'}
=\Psi_{ABCD}\varepsilon_{C'D'}-\frac{1}{3}\Lambda\varepsilon_{A(C}\varepsilon
_{D)B}\varepsilon_{C'D'}$, where, using Einstein's equations, we have already 
taken into account the vanishing of the trace-free part of the Ricci tensor. 
In terms of these quantities, the integrability condition (\ref{eq:3.3}) can 
be rewritten as 

\begin{equation}
0=\Bigl(-R_{ABef}+\frac{1}{3}\Lambda\varepsilon_{A(E}\varepsilon_{F)B}
\varepsilon_{E'F'}\Bigr)P^e_cP^f_d\lambda^B=\lambda^B\Psi_{ABEF}\varepsilon
_{E'F'}P^e_cP^f_d. \label{eq:3.4}
\end{equation}
Now we show that the whole Weyl spinor is also vanishing. 

For, let us introduce the complex null vectors $M^a:=\lambda^At^{A'}{}_B
\lambda^B$ and $\bar M^a:=\bar\lambda^{A'}t^A{}_{B'}\bar\lambda^{B'}$, which are 
tangent to $\Sigma$. These vectors together with $Z^a:=P^a_b\lambda^B\bar
\lambda^{B'}$ form a basis on an open dense subset of $\Sigma$ (namely on the 
subset where $\lambda^A$ is non-zero). Clearly, these satisfy the orthogonality 
and normalization conditions $Z_aM^a=M_aM^a=0$ and $Z_aZ^a=M_a\bar M^a=-(t_{AA'}
\lambda^A\bar\lambda^{A'})^2=:-\vert Z\vert^2$, respectively, and $Z^a=\frac{1}
{2}(\lambda^A\bar\lambda^{A'}-2\vert Z\vert^2t^A{}_{B'}\bar\lambda^{B'}t^{A'}{}_B
\lambda^B)$ holds. Then contracting (\ref{eq:3.4}), respectively, with $M^c
\bar M^d$, $Z^cM^d$ and $Z^c\bar M^d$, we obtain the vanishing of $\Psi_{ABCD}
\lambda^B\lambda^Ct^D{}_{D'}\bar\lambda^{D'}$, $\Psi_{ABCD}\lambda^B\lambda^C
\lambda^D$ and $\Psi_{ABCD}\lambda^Bt^C{}_{C'}\bar\lambda^{C'}t^D{}_{D'}\bar
\lambda^{D'}$. These imply that $\Psi_{ABCD}=\Psi\lambda_A\lambda_B\lambda_C
\lambda_D$ for some complex function $\Psi$ on $\Sigma$. Repeating this 
argumentation with the spinor field $\mu^A$, we obtain that $\Psi_{ABCD}=\tilde
\Psi\mu_A\mu_B\mu_C\mu_D$ for some $\tilde\Psi$. Since, however, the spinor 
fields $\lambda^A$ and $\mu^A$ are not proportional to each other on any open 
subset of $\Sigma$, this yields that $\Psi=\tilde\Psi=0$. Therefore, at the 
points of $\Sigma$ the curvature tensor of the spacetime is that of the de 
Sitter spacetime. 

Finally, let us foliate an open neighbourhood of $\Sigma$ in the spacetime by 
the spacelike hypersurfaces $\Sigma_t$ obtained by Lie dragging $\Sigma_0:=
\Sigma$ along its own unit timelike normal $t^a$. Then, the Bianchi identities, 
written in their 3+1 form with respect to this foliation by Friedrich 
\cite{Helmut}, yield that the spacetime curvature tensor is that of the de 
Sitter spacetime on the neighbouring leaves of the foliation. Hence the 
spacetime is locally isometric to the de Sitter spacetime with the 
cosmological constant $\Lambda$ in which the typical Cauchy surface is 
homeomorphic to $\Sigma$ (see \cite{Budic:etal}). 
\end{proof}

This result is slightly weaker than the analogous one in \cite{Sz12,Sz13} for 
the zero cosmological constant case, because there the topology of the 3-space 
has also been determined. In fact, there is a uniqueness for the topology of 
the holonomically trivial compact 3-spaces: they are necessarily tori (see 
\cite{KN1}, Theorem 4.2 in Chapter V, pp. 211-221), while, as we already noted 
above, the topology of 3-spaces with constant positive curvature is far from 
being fixed.  

In both the $\Lambda=0$ and $\Lambda>0$ cases the first eigenvalue of $2{\cal 
D}^{AA'}{\cal D}_{A'B}$ in the minimal total mass density configurations is 
built only from the curvature of ${\cal D}_e$, which has a uniform value in 
the 3-space. Thus, to give a more detailed characterization of the 3-spaces 
we should consider other eigenvalues. In the spatially flat, closed Bianchi I. 
model, the first eigenvalue of the Riemannian Dirac operator is zero, and the 
three parameters specifying the `size' of the flat 3-space are encoded in the 
higher eigenvalues (see \cite{Sz13}). Similarly, in the locally de Sitter 
case, the volume of the 3-space is expected to be recoverable from the higher 
eigenvalues. If this were indeed the case, then in terms of the curvature and 
the volume the \emph{topology} of the 3-space could also be characterized, at 
least partly, by the eigenvalues of $2{\cal D}^{AA'}{\cal D}_{A'B}$. To see 
this, observe that the local isometry $\pi:\tilde\Sigma_t\rightarrow\Sigma_t$ 
is a {\emph{universal covering}, because $\tilde\Sigma_t\approx S^3$ is simply 
connected. Since $\pi$ is smooth and $\tilde\Sigma_t$ is compact, it is a 
proper map, i.e. $\pi^{-1}(K)\subset\tilde\Sigma_t$ is compact for any compact 
$K\subset\Sigma_t$. Hence $\int_{\tilde\Sigma_t}\pi^*(\omega)_{abc}=\deg(\pi)\int
_{\Sigma_t}\omega_{abc}$ holds for any 3-form $\omega_{abc}$ on $\Sigma_t$, where 
$\deg(\pi)\in\mathbb{Z}$ is the degree of $\pi$ (see e.g. \cite{Spivak1}, p. 
275). Thus, in particular, ${\rm vol}(\tilde\Sigma_t)=\deg(\pi)\,{\rm vol}
(\Sigma_t)$. Actually, $\deg(\pi)$ is positive because $\pi$ is orientation 
preserving, and $\deg(\pi)$ can be interpreted as the number of how many times 
$\tilde\Sigma_t$ covers $\Sigma_t$. Therefore, since the curvature determines 
the volume of $\tilde\Sigma_t\approx S^3$, by ${\rm vol}(\tilde\Sigma_t)/{\rm 
vol}(\Sigma_t)=\deg(\pi)$ the curvature and the volume of the 3-space $\Sigma
_t$ determine $\deg(\pi)$, i.e. the number how many times $S^3$ covers the 
typical Cauchy surface $\Sigma$.

%%%%%%%%%%%%%%%%%%%%%%%%%%%%%%%%%%%%%%%%%%%%%%%%%%%%%%%%%%%%

\section{Total mass density in Bianchi A models}
\label{sec-4}

In Bianchi A cosmological spacetimes by the spatial homogeneity the eigenvalue 
equation for the Sen--Witten operator always has two independent eigenspinors 
whose \emph{components} in the \emph{left-invariant} frame are constant. In 
particular, on the $t={\rm const}$ spacelike hypersurfaces of the $k=1$ 
(closed) FRW spacetime these spinor fields are just the \emph{minimizer spinor 
fields} of ${\tt M}$. (See the discussions following equations (\ref{eq:A.2.7}) 
and (\ref{eq:A.2.10}) of Appendix \ref{sub-A.2}.) With this spinor field the 
integrand 

\begin{equation}
\frac{4}{\kappa}t^{AA'}t^{BB'}t^{CC'}{\cal D}_{(AB}\lambda_{C)}{\cal D}_{(A'B'}\bar
\lambda_{C')}+t^a\bigl(T_{ab}+\frac{\Lambda}{\kappa}g_{ab}\bigr)\lambda^B\bar
\lambda^{B'} \label{eq:4.1}
\end{equation}
in the definition (\ref{eq:2.4}) of ${\tt M}$ is \emph{constant} on the 
3-spaces $\Sigma_t$ of homogeneity. Thus, even though the integral of 
(\ref{eq:4.1}) does not exist on non-compact 3-spaces of homogeneity, this 
expression can be used to \emph{define} the total mass density by fixing the 
scale of the spinor field by $t_{AA'}\lambda^A\bar\lambda^{A'}=1$ an any given 
point of $\Sigma_t$. By equation (\ref{eq:A.2.6c}) of Appendix \ref{sub-A.2} 
the 3-surface twistor term of (\ref{eq:4.1}) in general is non-trivial, and 
could be interpreted as the contribution of the gravitational `field' to the 
total mass density of the matter+gravity system.

%%%%%%%%%%%%%%%%%%%%%%%%%%%%%%%%%%%%%%%%%%%%%%%%%%%%%%%%%%%

\section{Appendix}
\label{sec-A}

\subsection{An application of Aronszajn's theorem to eigenspinors}
\label{sub-A.1}

In this appendix we prove the following statement: 

\begin{theorem}
If $\lambda^A$ is an eigenspinor of $2{\cal D}^{AA'}{\cal D}_{A'B}$ on 
$\Sigma$, then it cannot have any zero of infinite order in the 1--mean (for 
the definition see below). 
\end{theorem}

\medskip
\noindent
Let $o\in\Sigma$, $r>0$ and let $B(o,r)$ denote the set of points of $\Sigma$ 
whose geodesic distance from $o$ is less than $r$. Following Aronszajn 
\cite{Aronszajn}, we say that $\lambda^A$ has a zero of order $\omega$ at $o$ 
in the $p$--mean if $\int_{B(o,r)}\vert\lambda^A\vert^p{\rm d}\Sigma=O(r
^{\omega p+3})$, where $\vert\lambda^A\vert^2:=\sqrt{2}t_{AA'}\lambda^A\bar
\lambda^{A'}$, the square of the pointwise norm of the spinor field. The zero 
is said to be of infinite order if it is of order $\omega$ for all $\omega>0$. 
Clearly, if $\lambda^A$ is vanishing on an open subset $O\subset\Sigma$, then 
any $o\in O$ is a zero of infinite order. Hence, the theorem above guarantees 
that no eigenspinor of $2{\cal D}^{AA'}{\cal D}_{A'B}$ can be vanishing on any 
open subset of $\Sigma$.

\begin{proof}
In \cite{Sz12} we derived a Lichnerowicz type identity for any spinor field 
(equation (2.11) of \cite{Sz12}), which, by the expression for the constraint 
parts of the spacetime Einstein tensor (equations (2.9) and (2.10) of 
\cite{Sz12}) and for the Sen derivative operator in terms of the intrinsic 
Levi-Civita derivative and the extrinsic curvature (equations (2.2) and (2.4) 
of \cite{Sz12}), takes the form 

\begin{equation}
2{\cal D}^{AA'}{\cal D}_{A'B}\lambda^B=D_eD^e\lambda^A+\bigl(D^{AA'}\chi\bigr)
t_{A'B}\lambda^B+\frac{1}{4}\bigl({\cal R}+\chi^2\bigr)\lambda^A. 
\label{eq:A.1.1}
\end{equation}
Thus if $\lambda^A$ is an eigenspinor of $2{\cal D}^{AA'}{\cal D}_{A'B}$ with 
the eigenvalue $\alpha^2$, then $\lambda^A$ satisfies 

\begin{equation}
D_eD^e\lambda^A=\bigl(D^A{}_{A'}\chi\bigr)t^{A'}{}_B\lambda^B+\frac{1}{4}\bigl(
4\alpha^2-{\cal R}-\chi^2\bigr)\lambda^A. \label{eq:A.1.2}
\end{equation}
Let $U\subset\Sigma$ be an open set with \emph{compact} closure which is the 
domain of the local coordinate system $\{x^\alpha\}$, and fix a normalized dual 
spin frame field $\{\varepsilon^A_{\uA},\varepsilon^{\uA}_A\}$ on $U$. Then the 
coordinate Laplacian of the \emph{components} of the spinor field has the 
structure 

\begin{equation*}
h^{\alpha\beta}\partial_\alpha\partial_\beta\lambda^{\uA}=\bigl(D_eD^e\lambda^A
\bigr)\varepsilon^{\uA}_A+E^{\alpha{\uA}}{}_{\uB}\partial_\alpha\lambda^{\uB}
+E^{\uA}{}_{\uB}\lambda^{\uB},
\end{equation*}
where the coefficients $E^{\alpha{\uA}}{}_{\uB}$ and $E^{\uA}{}_{\uB}$ are built 
from the Christoffel symbols and the components of the connection 1-form on 
the spinor bundle in the spinor basis above. Taking into account 
(\ref{eq:A.1.2}) and writing its right hand side as $F^A{}_B\lambda^B$, we 
obtain 

\begin{equation}
\vert h^{\alpha\beta}\partial_\alpha\partial_\beta\lambda^{\uA}\vert\leq\vert 
E^{\alpha{\uA}}{}_{\uB}\partial_\alpha\lambda^{\uB}\vert+\vert E^{\uA}{}_{\uB}
\lambda^{\uB}\vert+\vert F^{\uA}{}_{\uB}\lambda^{\uB}\vert, \label{eq:A.1.4}
\end{equation}
where $\vert\,.\,\vert$ means pointwise \emph{absolute value} of spinor/tensor 
\emph{components}. Next, recall that on any finite dimensional complex vector 
space, say $\mathbb{C}^n$, any two norms $\Vert x\Vert_p:=(\sum_{k=1}^n\vert x
^k\vert^p)^{1/p}$, $x=\{x^k\}\in\mathbb{C}^n$, are equivalent, where $1\leq p<
\infty$. In particular, $\Vert x\Vert_1\leq n\Vert x\Vert_2$ holds. Thus, if 
$x,y,z\in\mathbb{C}^n$, then with the notation $w:=(x,y,z)\in\mathbb{C}^{3n}$ 
we have that $\Vert x\Vert_1+\Vert y\Vert_1+\Vert z\Vert_1=\Vert w\Vert_1
\leq3n\Vert w\Vert_2$. Applying this inequality to the pointwise norm on the 
right of (\ref{eq:A.1.4}), we obtain 

\begin{equation*}
\vert h^{\alpha\beta}\partial_\alpha\partial_\beta\lambda^{\uA}\vert^2\leq 36\Bigl(
\vert E^{\alpha{\uA}}{}_{\uB}\partial_\alpha\lambda^{\uB}\vert^2+\vert E^{\uA}{}_{\uB}
\lambda^{\uB}\vert^2+\vert F^{\uA}{}_{\uB}\lambda^{\uB}\vert^2\Bigr).
\end{equation*}
Since the geometry is smooth and $U$ is of compact closure, there exists a 
positive constant $C$ such that for both ${\uA}=0,1$ 

\begin{equation*}
\vert h^{\alpha\beta}\partial_\alpha\partial_\beta\lambda^{\uA}\vert^2\leq C\,\delta
_{{\uB}{\uB}'}\Bigl(-h^{\alpha\beta}(\partial_\alpha\lambda^{\uB})(\partial_\beta
\bar\lambda^{{\uB}'})+\lambda^{\uB}\bar\lambda^{{\uB}'}\Bigr) 
\end{equation*}
holds, where $\delta_{{\uB}{\uB}'}$ is the Kronecker delta. However, this is 
just the condition of Aronszajn's theorem \cite{Aronszajn} (see Remark 3 in 
Aronszajn's paper), which guarantees the vanishing of $\lambda^A$ on $U$ if 
it has a zero of infinite order somewhere in $U$. Finally, covering $\Sigma$ 
by such coordinate domains $U$ and applying the above result to the 
overlapping domains, we find that $\lambda^A$ is vanishing on the whole of 
$\Sigma$. 
\end{proof} 

%%%%%%%%%%%%%%%%%%%%%%%%%%%%%%%%%%%%%%%%%%%%%%%%%%%%%%%%%

\subsection{The ${\cal D}$ and ${\cal T}$ operators in Bianchi 3-spaces}
\label{sub-A.2}

\subsubsection{The geometry of the Bianchi cosmological spacetimes}
\label{sub-A.2.1}

Let the spacetime be a homogeneous Bianchi cosmological spacetime, foliated 
by the spacelike hypersurfaces $\Sigma_t$ which are the transitivity surfaces 
of the isometry group (see e.g. \cite{Wald}). Let $\{e^a_{\bi},\vartheta^{\bi}
_a\}$, ${\bi}=1,2,3$, be a globally defined $h_{ab}$-orthonormal dual frame 
field on $\Sigma_0$, which is \emph{left invariant} with respect to the 
(simply transitive) action of the isometry group $G$ on $\Sigma_0$. Then the 
structure constants of $G$ can be given in this basis, too, by $c^{\bi}_{\bj\bk}
:=[e_{\bj},e_{\bk}]^a\vartheta^{\bi}_a$, and $D_{[a}\vartheta^{\bi}_{b]}=-\frac{1}
{2}c^{\bi}_{\bj\bk}\vartheta^{\bj}_{[a}\vartheta^{\bk}_{b]}$ also holds. 
Following \cite{Wald}, we parametrize the structure constants as $c^{\bi}
_{\bj\bk}=M^{\bi\bm}\varepsilon_{\bm\bj\bk}+\delta^{\bi}_{[{\bj}}A_{\bk]}$, where 
$M^{\bi\bj}$ is a symmetric real matrix, $M^{\bi\bj}A_{\bj}=0$ holds and 
$\varepsilon_{\bk\bl\bm}$ is the alternating Levi-Civita symbol. The group $G$ 
is said to belong to Bianchi class A if $A_{\bi}=0$, otherwise it belongs to 
class B. Extending the dual frame field $\{e^a_{\bi},\vartheta^{\bi}_a\}$ along 
the timelike geodesic normals to the other leaves $\Sigma_t$ we obtain a 
globally defined frame field which provides a convenient `background' to 
describe the dynamics, but, in general, it will \emph{not} be 
$h_{ab}$-orthonormal on the leaves other than $\Sigma_0$. In this basis $h_{ab}
=h_{\bi\bj}\vartheta^{\bi}_a\vartheta^{\bj}_b$, where the components $h_{\bi\bj}$ 
are functions only of the time coordinate $t$. 

Thus, let us write the vectors of an $h_{ab}$-orthonormal basis as $E^a_{\bi}:=
e^a_{\bj}\Phi^{\bj}{}_{\bi}$, and similarly $\theta^{\bi}_a:=\vartheta^{\bj}_a
\Phi_{\bj}{}^{\bi}$, where the matrices $\Phi^{\bi}{}_{\bj}$, $\Phi_{\bj}{}^{\bi}$ 
depend only on $t$, have positive determinant, satisfy $\Phi^{\bi}{}_{\bj}
\Phi_{\bk}{}^{\bj}=\delta^{\bi}_{\bk}$ and $h_{\bi\bj}=\Phi_{\bi}{}^{\bk}\Phi_{\bj}
{}^{\bl}\eta_{\bk\bl}$, and at $t=0$ they reduce to the unit matrix $\delta^{\bi}
_{\bj}$. Clearly, $\{E^a_{\bi},\theta^{\bi}_a\}$ is also a left-invariant dual 
basis, but $E^a_{\bi}=h^{ab}\theta^{\bj}_b\eta_{\bj\bi}$ also holds, where boldface 
(i.e. frame name) index lowering and raising are defined by $\eta_{\bi\bj}:=-
\delta_{\bi\bj}$ and its inverse. (Note, however, that $h^{\bi\bj}$ denotes the 
\emph{inverse} of the matrix $h_{\bi\bj}$.) The structure `constants' of this 
frame are $C^{\bi}_{\bj\bk}:=[E_{\bj},E_{\bk}]^a\theta_a^{\bi}=c^{\bl}_{\bm\bn}\Phi
_{\bl}{}^{\bi}\Phi^{\bm}{}_{\bj}\Phi^{\bn}{}_{\bk}$, which depend on the time 
coordinate. It is a simple calculation to show that these structure constants 
determine completely the Ricci rotation coefficients of the spatial 
Levi-Civita connection in the frame $\{E^a_{\bi},\theta^{\bi}_a\}$: 

\begin{equation}
\gamma^{\bi}_{\bk\bj}:=E^e_{\bk}\gamma^{\bi}_{e{\bj}}:=\theta^{\bi}_aE^e_{\bk}\bigl(
D_eE^a_{\bj}\bigr)=\frac{1}{2}\Bigl(C^{\bi}_{\bk\bj}+\eta^{\bi\bm}C^{\bn}_{\bm\bk}
\eta_{\bn\bj}+\eta^{\bi\bm}C^{\bn}_{\bm\bj}\eta_{\bn\bk}\Bigr). \label{eq:A.2.1}
\end{equation}
Clearly, these are constant on each $\Sigma_t$, and, with the parametrization 
of $c^{\bi}_{\bj\bk}$ above, we have that 

\begin{equation}
\frac{1}{2}\gamma^{\bi\bj}_{\bk}\varepsilon_{\bi\bj\bl}=-\frac{1}{\sqrt{\vert h
\vert}}M^{\bm\bn}\Phi_{\bm}{}^{\bi}\Phi_{\bn}{}^{\bj}\eta_{\bi\bk}\eta_{\bj\bl}+
\frac{1}{2}\frac{1}{\sqrt{\vert h\vert}}\eta_{\bk\bl}M^{\bi\bj}h_{\bi\bj}-
\frac{1}{2}A_{\bi}\Phi^{\bi}{}_{\bj}\varepsilon^{\bj}{}_{\bk\bl}, \label{eq:A.2.2}
\end{equation}
where $\vert h\vert:=\vert \det(h_{\bi\bj})\vert=(\det(\Phi^{\bi}{}_{\bj}))^{-2}$. 
Equation (\ref{eq:A.2.2}) implies, in particular, that the frame field 
$\{E^a_{\bi}\}$ satisfies Nester's special orthonormal frame gauge condition on 
$\Sigma_t$ \cite{Nester}: Nester requires $q:=\gamma^{\bi\bj}_{\bk}\varepsilon
_{\bi\bj}{}^{\bk}$ to be constant and $q_{\bi}\theta^{\bi}_a$ to be closed, where 
$q_{\bj}:=\gamma^{\bi}_{\bi\bj}$. By $q=M^{\bi\bj}h_{\bi\bj}/\sqrt{\vert h\vert}$ 
the first condition is clearly satisfied, and by $q_{\bi}=A_{\bj}\Phi^{\bj}{}
_{\bi}$, $M^{\bi\bj}A_{\bj}=0$ and $D_{[a}\theta^{\bi}_{b]}=-\frac{1}{2}C^{\bi}
_{\bj\bk}\theta^{\bj}_{[a}\theta^{\bk}_{b]}$ the 1-form $q_{\bi}\theta^{\bi}_a$ is, 
indeed, closed. Since the Ricci rotation coefficients are constants in the 
left-invariant frame, it is straightforward to compute the components of the 
curvature tensor and its scalar curvature. For the latter, we obtain 

\begin{eqnarray}
{\cal R}\!\!\!\!&=\!\!\!\!&-\frac{1}{2}\Bigl(M^{\bk\bl}M^{\bm\bn}\eta_{\bk\bm}
  \eta_{\bl\bn}\eta_{\bi\bj}-2\eta_{\bi\bk}M^{\bk\bm}\eta_{\bm\bn}M^{\bn\bl}\eta
  _{\bl\bj}+2M^{\bm\bn}\eta_{\bm\bn}M^{\bk\bl}\eta_{\bk\bi}\eta_{\bl\bj}-\nonumber \\
\!\!\!\!&-\!\!\!\!&\bigl(M^{\bk\bl}\eta_{\bk\bl}\bigr)^2\eta_{\bi\bj}\Bigr)
  h^{\bi\bj}-\frac{1}{2\vert h\vert}M^{\bk\bi}M^{\bl\bj}h_{\bk\bl}h_{\bi\bj}+
  \frac{3}{2}A_{\bi}A_{\bj}h^{\bi\bj},  \label{eq:A.2.3}
\end{eqnarray}
which is constant on $\Sigma_t$, but depends on the time coordinate $t$ since 
$h_{\bi\bj}$ does.

\subsubsection{The ${\cal D}$ and ${\cal T}$ operators}
\label{sub-A.2.2}

If $\{\varepsilon^A_{\uA},\varepsilon^{\uA}_A\}$, ${\uA}=0,1$, is a normalized 
spinor basis associated with the dual frame field $\{E^a_{\bi},\theta^{\bi}_a
\}$, then the spinor connection 1-forms of the intrinsic Levi-Civita covariant 
derivative on the spinor bundle are given by 

\begin{equation}
\gamma^{\uA}_{e{\uB}}=-\frac{\rm i}{2\sqrt{2}}\gamma^{\bi\bj}_e\varepsilon_{\bi\bj}
{}^{\bk}\sigma^{\uA}_{{\bk}{\uB}}, \label{eq:A.2.4}
\end{equation}
where $\sigma^{\uA}_{{\bk}{\uB}}$ are the standard $SU(2)$ Pauli matrices 
(including the factor $1/\sqrt{2}$). Then, by contracting it with $E^e_{\bi}
\sigma^{\bi}_{\uC\uD}$, it is straightforward to compute its various irreducible 
parts. We obtain 

\begin{equation}
\gamma^{\uD}_{\uD\uA\uB}=\frac{1}{2}A_{\bi}\Phi^{\bi}{}_{\bj}\sigma^{\bj}_{\uA\uB}-
\frac{\rm i}{4\sqrt{2\vert h\vert}}M^{\bi\bj}h_{\bi\bj}\varepsilon_{\uA\uB}, 
\hskip 15pt
\gamma^{\uA}_{(\uB\uC\uD)}=\frac{\rm i}{2\sqrt{2}\vert h\vert}M^{\bi\bj}\Phi_{\bi}
{}^{\bk}\Phi_{\bj}{}^{\bl}\eta_{\bl\bm}\sigma^{\uA}_{{\bk}({\uB}}\sigma^{\bm}_{\uC\uD)}; 
\label{eq:A.2.5} 
\end{equation}
i.e. the irreducible parts of the first are proportional to $A_{\bi}$ and the 
\emph{trace} of $M^{\bi\bj}$ with respect to $h_{\bi\bj}$, respectively, while 
the second to the $h_{\bi\bj}$--\emph{trace-free part} of $M^{\bi\bj}$. 

By (\ref{eq:A.2.5}) the Riemannian Dirac and 3-surface twistor operators, 
$D_{AB}\lambda^B$ and $T(\lambda)_{ABC}$ $:=D_{(AB}\lambda_{C)}$ on $\Sigma_t$, 
respectively, take the form 

\begin{eqnarray}
\varepsilon^A_{\uA}D_{AB}\lambda^B\!\!\!\!&=\!\!\!\!&\sigma^{\bi}_{\uA\uB}
 \Bigl(E^e_{\bi}\bigl(\partial_e\lambda^{\uB}\bigr)+\frac{1}{2}A_{\bj}\Phi^{\bj}
 {}_{\bi}\lambda^{\uB}\Bigr)+\frac{\rm i}{4\sqrt{2}}\frac{1}{\sqrt{\vert h
 \vert}}M^{\bi\bj}h_{\bi\bj}\lambda_{\uA}, \label{eq:A.2.6a}\\
T(\lambda)_{\uA\uB\uC}\!\!\!\!&=\!\!\!\!&E^e_{\bi}\bigl(\partial_e\lambda
 _{({\uA}}\bigr)\sigma^{\bi}_{{\uB\uC})}-\frac{\rm i}{2\sqrt{2}\sqrt{\vert h\vert}}
 M^{\bm\bn}\Phi_{\bm}{}^{\bi}\Phi_{\bn}{}^{\bj}\eta_{\bj\bk}\sigma^{\uD}_{{\bi}({\uA}}
 \sigma^{\bk}_{{\uB\uC})}\lambda_{\uD}. \label{eq:A.2.6b}
\end{eqnarray}
(\ref{eq:A.2.6a}) shows that spinor fields with \emph{constant components} in 
the left-invariant frame field are eigenspinors of the Riemannian Dirac 
operator precisely when $A_{\bi}=0$, i.e. when the group $G$ belongs to Bianchi 
class A, in which case the eigenvalue is $\beta=-M^{\bi\bj}h_{\bi\bj}/
(4\sqrt{\vert h\vert})$. The action of the Riemannian 3-surface twistor 
operator on \emph{such} a spinor field is 

\begin{equation}
T(\lambda)^{\uA\uB\uC}=-\frac{\rm i}{2\sqrt{2}\sqrt{\vert h\vert}}\Bigl(
\widehat{M}^{\bi\bj}\Phi_{\bi}{}^{\bm}\Phi_{\bj}{}^{\bn}+\frac{1}{3}M^{\bi\bj}
\eta_{\bi\bj}\widehat{h}^{\bm\bn}\Bigr)\sigma^{({\uA\uB}}_{\bm}\sigma^{{\uC})
\uD}_{\bn}\lambda_{\uD}, \label{eq:A.2.6c}
\end{equation}
where $\widehat{M}^{\bi\bj}:=M^{\bi\bj}-\frac{1}{3}M^{\bk\bl}\eta_{\bk\bl}\eta
^{\bi\bj}$ and $\widehat{h}^{\bi\bj}:=h^{\bi\bj}-\frac{1}{3}h^{\bk\bl}\eta_{\bk\bl}
\eta^{\bi\bj}$, the trace-free part of $M^{\bi\bj}$ and $h^{\bi\bj}$, 
respectively, with respect to $\eta_{\bi\bj}$. This is vanishing if $M^{\bi\bj}
=0$ (Bianchi I. models), or if $\widehat{M}^{\bi\bj}=0$ and the evolution of 
the spatial metric yields that $h_{\bi\bj}=S^2\eta_{\bi\bj}$ for some positive 
function $S=S(t)$ (closed FRW spacetimes in Bianchi IX). 

The action of the square of the Riemannian Dirac operator on the spinor field 
whose components $\lambda^{\uA}$ in the left-invariant frame are constant is 

\begin{equation}
\varepsilon^{\uA}_AD^{AB}D_{BC}\lambda^C=\frac{1}{8}\Bigl(\bigl(q_{\bi}q_{\bj}
\eta^{\bi\bj}+\frac{1}{4}q^2\bigr)\delta^{\uA}_{\uB}-{\rm i}\sqrt{2}qq^{\bi}
\sigma^{\uA}_{{\bi}{\uB}}\Bigr)\lambda^{\uB}, \label{eq:A.2.7}
\end{equation}
where, for the sake of brevity, we used the quantities $q$ and $q_{\bi}$ 
introduced in connection with Nester's gauge condition above. Therefore, 
\emph{such} a spinor field can be an eigenspinor of $2D^{AB}D_{BC}$ with the 
eigenvalue $\beta^2$ precisely when $qq_{\bi}=0$, and the eigenvalue is $\beta
^2=\frac{1}{4}q_{\bi}q_{\bj}\eta^{\bi\bj}+\frac{1}{16}q^2$. In particular, if 
$A_{\bi}=0$, $M^{\bi\bj}=\frac{1}{3}M\eta^{\bi\bj}$ and $h_{\bi\bj}=S^2\eta_{\bi\bj}$ 
(e.g. for the closed FRW spacetimes), then by (\ref{eq:A.2.3}) this yields 
that $\beta^2=\frac{3}{8}{\cal R}$, which is known to be just the 
\emph{smallest} eigenvalue of $2D^{AB}D_{BC}$ on the metric 3-spheres. 

If the symmetry group $G$ belongs to Bianchi class B, then no eigenspinor of 
the Riemannian Dirac operator has constant components in the left invariant 
frame field. Nevertheless, (\ref{eq:A.2.6a}) motivates us to consider the 
weaker condition 

\begin{equation}
E^e_{\bi}(\partial_e\lambda^{\uA})+\frac{\rm i}{2}q_{\bi}\lambda^{\uA}=0.
\label{eq:A.2.8}
\end{equation}
It is a straightforward calculation to show that its integrability conditions 
are satisfied \emph{identically}, and hence this equation \emph{locally} 
admits precisely two linearly independent solutions, which are specified 
completely by their own value at an arbitrary point of $\Sigma_t$. Let these 
two solutions be $\lambda^{\uA}$ and $\mu^{\uA}$, and form their symplectic 
scalar product $\omega:=\lambda^{\uA}\mu^{\uB}\varepsilon_{\uA\uB}$. Then 
(\ref{eq:A.2.8}) yields that $q_{\bi}\theta^{\bi}_e={\rm i}\partial_e(\ln
\omega)$, i.e., if (\ref{eq:A.2.8}) admits \emph{global} solutions, then the 
combination $q_{\bi}\theta^{\bi}_e$ of the left invariant orthonormal 1-form 
basis would be \emph{exact}. Then this global solution would be an eigenspinor 
of the Riemannian Dirac operator also with the eigenvalue $\beta=-M^{\bi\bj}
h_{\bi\bj}/(4\sqrt{\vert h\vert})$. 

The unitary spinor form of the derivative operator of the Sen connection is 
well known to be ${\cal D}_{AB}\lambda_C=D_{AB}\lambda_C-\frac{1}{\sqrt{2}}\chi
_{ABCD}\lambda^D$, where $\chi_{ABCD}:=2t_B{}^{A'}t_D{}^{C'}\chi_{AA'CC'}$. Hence, 
if the components $\chi_{\bi\bj}$ of the extrinsic curvature are defined by 
$\chi_{ab}=:\chi_{\bi\bj}\vartheta^{\bi}_a\vartheta^{\bj}_b$, which are also 
constant on $\Sigma_t$ but in general depend on $t$, then the Sen--Witten and 
the 3-surface twistor operators, respectively, take the form 

\begin{eqnarray}
\varepsilon^A_{\uA}{\cal D}_{AB}\lambda^B\!\!\!\!&=\!\!\!\!&\varepsilon^A
  _{\uA}D_{AB}\lambda^B+\frac{1}{2\sqrt{2}}\chi\lambda_{\uA}, 
  \label{eq:A.2.9a} \\
{\cal T}(\lambda)_{\uA\uB\uC}\!\!\!\!&=\!\!\!\!&T(\lambda)_{\uA\uB\uC}-\frac{1}
  {\sqrt{2}}\chi_{\bi\bj}\Phi^{\bi}{}_{\bk}\Phi^{\bj}{}_{\bl}\sigma^{\bk}_{({\uA\uB}}
  \sigma^{\bl}_{{\uC}){\uD}}\lambda^{\uD}. \label{eq:A.2.9b}
\end{eqnarray}
Thus, if the extrinsic curvature is a pure trace, $\chi_{ab}=\frac{1}{3}\chi 
h_{ab}$, i.e. $\chi_{\bk\bl}\Phi^{\bk}{}_{\bi}\Phi^{\bl}{}_{\bj}=\frac{1}{3}\chi
\eta_{\bi\bj}$ (as in the FRW spacetimes), then by (\ref{eq:A.2.9b}) the 
3-surface twistor and the Riemannian 3-surface twistor operators coincide. 
Since in general $2{\cal D}^{AA'}{\cal D}_{A'B}\lambda^A=2D^{AA'}D_{A'B}\lambda^B
+\lambda^Bt_{BB'}(D^{B'A}\chi)+\frac{1}{4}\chi^2\lambda^A$, in spatially 
homogeneous spacetimes, where $\chi$ is spatially constant, we have for any 
eigenspinor $\lambda^A$ of the square of the Riemannian Dirac operator with 
eigenvalue $\beta^2$ that 

\begin{equation}
2{\cal D}^{AA'}{\cal D}_{A'B}\lambda^B=2D^{AA'}D_{A'B}\lambda^B+\frac{1}{4}
\chi^2\lambda^A=\bigl(\beta^2+\frac{1}{4}\chi^2\bigr)\lambda^A. 
\label{eq:A.2.10}
\end{equation}
Thus, $\lambda^A$ will be an eigenspinor of the square of the Sen--Witten 
operator, too, with the eigenvalue $\beta^2+\frac{1}{4}\chi^2$. In particular, 
in the closed FRW spacetimes the first eigenvalue of $2{\cal D}^{AA'}{\cal D}
_{A'B}$ is $\frac{3}{8}{\cal R}+\frac{1}{4}\chi^2$, and the corresponding 
eigenspinors are just the ones whose components in the left-invariant frame 
field are constant.

%%%%%%%%%%%%%%%%%%%%%%%%%%%%%%%%%%%%%%%%%

\end{document}